\documentclass[conference,onecolumn, draftcls]{IEEETran}

\usepackage{amsfonts}
\usepackage{amsmath}
\usepackage{amssymb}
\usepackage{amsthm}
\usepackage{xcolor}
\usepackage{bbm}
\usepackage{tikzscale}
\usepackage[ruled,norelsize]{algorithm2e}
\usepackage{epstopdf}
\usepackage{tikz}
\usetikzlibrary{calc,shapes,snakes}
\newtheorem{thm}{Theorem}

\newtheorem{prb}{Problem}

\newtheorem{cor}{Corollary}

\newcommand{\ie}{{\it i.e.},\ }

\makeatletter
\newcommand{\removelatexerror}{\let\@latex@error\@gobble}
\makeatother

\begin{document}
\title{Distributed Cooperative Caching for VoD with Geographic Constraints}

\author{
\IEEEauthorblockN{Konstantin Avrachenkov}
\IEEEauthorblockA{INRIA\\ Sophia Antipolis, France\\
konstantin.avrachenkov@inria.fr}
\and
\IEEEauthorblockN{Jasper Goseling}
\IEEEauthorblockA{Mathematics of Operations Research\\University of Twente, The Netherlands\\
j.goseling@utwente.nl}
\and
\IEEEauthorblockN{Berksan Serbetci\IEEEauthorrefmark{1}}
\IEEEauthorblockA{EURECOM\\Sophia Antipolis, France\\
serbetci@eurecom.fr\\
\IEEEauthorrefmark{1}Corresponding author.}
}
\maketitle

%\keywords{Caching; Wireless networks; Distributed optimization; Game theory; Simulated annealing}
\maketitle

\begin{abstract}
We consider caching of video streams in a cellular network in which each base station is equipped with a cache. Video streams are partitioned into multiple substreams and the goal is to place substreams in caches such that the residual backhaul load is minimized. We consider two coding mechanisms for the substreams: Layered coding (LC) mechanism and Multiple Description Coding (MDC). We develop a distributed asynchronous algorithm for deciding which files to store in which cache to minimize the residual bandwidth, \ie the cost for downloading the missing substreams of the user's requested video with a certain video quality from the gateway (\ie the main server). We show that our algorithm converges rapidly. Also, we show that MDC partitioning is better than LC mechanism when the most popular content is stored in caches; however, our algorithm diminishes this difference.
\end{abstract}

\section{Introduction}
\label{sec:intro}
Data traffic in cellular networks is rapidly expanding and one of the main reasons is the explosive growth in the demand for Video-on-Demand (VoD) content. Recent studies show that existing network infrastructures will not be able to support the demand~\cite{cisco}. One of the bottlenecks will be formed by the backhaul links that connect base stations to the core network and we need to utilize these links as efficiently as possible. A promising means to efficiently handle this issue is to proactively cache data in the base stations. The idea is to store part of the popular data at the wireless edge and use the backhaul to refresh the stored data at off-peak hours since users' demand distribution over time is varying slowly. In order to embody this, content and network service providers are deploying Video Hub Offices (VHOs) serving as caches in metropolitan areas to serve VoD subscribers in that area. In this way, caches containing popular content serve as helpers to the overall system and decrease the backhaul load.

Our goal in this paper is on developing low-complexity distributed and asynchronous content placement algorithms for VoD. As content and network service providers would like to optimize the cached content in caches (\ie base stations, or more specifically VHOs) while keeping the communication in the network to a minimum, the problem is of practical relevance in cellular networks. Our low-complexity distributed and asynchronous algorithm is designed to optimize the content and keeping the communication between caches to a minimum as the caches exchange information only locally.

Let us first provide a brief discussion of the related works. Caching has been receiving a lot of attention in the literature lately. Here we provide a brief overview of the work on caching that is most closely related to the current paper. A general outline of a distributed caching architecture for wireless networks considering a network model in which a bipartite graph indicates how users are connected to base stations has been presented in~\cite{SGDMC13}. Works in~\cite{AAG14,ABG16,BG15,SG18} consider placement of base stations in the plane according to a stochastic geometry where the main aim is to find the optimal placement strategy for the network for various performance measures. In~\cite{AGS17} we have developed a distributed asynchronous algorithm that deals with the miss probability minimization by casting the problem into the framework of potential games for such the networks where caches are located at arbitrary locations, and showed that our algorithm converges to a Nash equilibrium, and in fact to the best Nash equilibrium in most practical scenarios. Finally in~\cite{AAGLR16}, authors present an approach for optimal content placement for a large-scale VoD by formulating the problem as a mixed integer problem. The main difference between~\cite{AGS17, AAGLR16} and the present work is the considered performance measures, the concept of video partitioning and the heterogeneity of the video popularities over the network. For a more thorough review of the existing caching work, the reader is strongly encouraged to refer to papers~\cite{AGS17, PITTC18}.

In the remainder of the introduction we will give an overview of the model and contributions. 

We consider a network with caches located at arbitrary locations in the plane. Caches know their own coverage area as well as the coverage areas of other caches that overlap with this region. There is a content catalog from which users request videos with a certain video quality according to a known probability distribution. Videos are partitioned into chunks (\ie substreams) by using two coding mechanisms: Layered coding (LC) mechanism~\cite{WSBL03} and Multiple Description Coding (MDC)~\cite{G01}. Each cache can store a limited number of chunks and the goal is to minimize the residual bandwidth, \ie the cost for downloading the missing video chunks of the users' requested videos with certain video qualities from the main server. We develop low-complexity asynchronous distributed cooperative content placement caching algorithms that require communication only between caches with overlapping coverage areas. In the basic algorithm, at each iteration a cache will selfishly update its cache content by minimizing the local residual bandwidth and by considering the content stored by neighbouring caches. We show that our basic algorithm can be formulated as providing the best response dynamics in a potential game. We show that our algorithm is guaranteed to converge to a Nash equilibrium in finite time. We analyze the structure of the best response dynamics and provide the closed-form optimal solutions for the content placement by solving the linear optimization problem. Finally, we illustrate our results by a number of numerical results with various video quality request distributions for real world network models.

To specify, our contributions are as follows:
\begin{itemize}
\item We provide the partitioning models for the two aforementioned coding mechanisms and formulate problems in order to minimize the residual bandwidth for both mechanisms;
\item We provide a distributed asynchronous algorithm for optimizing the content placement which can be interpreted as giving the best response dynamics in a potential game for both partitioning mechanisms;
\item We prove that the best response dynamics can be obtained as a solution of a linear optimization problem;
\item We prove that our algorithm converges in finite time, and in fact rapidly;
\item We evaluate our algorithm through numerical examples using the base station locations from a real wireless network for the cellular network topology. We study the residual bandwidth evolution on a real network numerically for both coding mechanisms and show that our distributed caching algorithm strongly suggests using the LC mechanism for video partitioning in VoD systems;
\end{itemize}

Let us outline the organization of the paper.
In Section~\ref{sec:model} we give the formal model and problem definitions. In Section~\ref{sec:potentialgame} we provide the game formulation of the problem, analyze the structures of the best response dynamics and Nash equilibria.In Section~\ref{sec:performance}, we present practical implementations of our low-complexity algorithms for various video quality distributions and show the resulting performances for a real network. In Section~\ref{sec:discussion} we conclude the paper with some discussions and provide an outlook on future research.
%%%%%%%%%%%%%%%%%%%%%%%%
%
%
%
%
%
%%%%%%%%%%%%%%%%%%%%%%%%
\section{Model and Problem Definition}
\label{sec:model}
We consider a network of $N$ base stations that are located in the plane $\mathbb{R}^2$. We will use the notation $[1:N]=\{1,\dots,N\}$ and $\Theta = \mathbb{P}\left([1:N]\right) \setminus \emptyset$, where $\mathbb{P}\left([1:N]\right)$ is the power set of $[1:N]$. We specify the geometric configuration of the network through $A_s$, $s \in \Theta$, which denotes the area of the plane that is covered only by the caches in subset $s$, namely $A_s = (\cap_{\ell \in s} \bar{A}_\ell) \cap (\cap_{\ell \not \in s} \bar{A}_\ell^c)$, where $\bar{A}_\ell$ is the complete coverage region of cache $\ell$.

As a special case we will consider the case that all base stations have the same circular coverage region with radius $r$. In this case we specify the location of each base station, with $x_m$ for the location of base station $m\in[1:N]$. We then obtain $\bar{A}_{m}$ as the disc of radius $r$ around $x_m$.

Each base station is equipped with a cache that can be used to store videos from a content library $\mathcal{C}_v = \left\{c_{1}, \dots, c_{J}\right\}$, where $J < \infty$. Each element $c_{j}$ represents a video, where $j$ is the file (video) index. The size of video $c_j$ is denoted by $w_{j}$.

Next, videos are partitioned into chunks. We assume that video $c_j$ is partitioned into $Q$ chunks and consequently, the chunk library consists of video chunks $\mathcal{C}_c = \left\{c_{1,1}, \dots, c_{1,Q},\dots, c_{J,1},\dots, c_{J,Q}\right\}$, where $J, Q < \infty$. Each element $c_{j,q}$ represents a video chunk, where $j$ is the file (video) index and $q$ is the chunk index. The chunk size for the element $c_{j,q}$ is denoted by $w_{j,q}$.

The motivation behind partitioning videos into chunks is as follows. In this work we will consider various coding mechanisms. If the video is partitioned by using layered coding (LC) mechanisms (e.g., MPEG-2, MPEG-4, etc.), there is a base layer and enhancement layers. The base layer is necessary for the media stream to be decoded. Accordingly, the further enhancement layers are applied to improve the video quality. Hence, for the layered coding mechanism, the element $c_{j,1}$ represents the base layer for video $j$, and $c_{j,2}$ is the first enhancement layer required to obtain a better video quality, and so forth. For the LC mechanism, $Q = Q_{LC}$ and the chunk size $w_{j,q}$ will then depend on the average base size of the different video qualities, and the total size of the video $j$ with the highest video quality is equal to
\begin{equation}
w_j = \sum_{q = 1}^{Q_{LC}} w_{j,q}. \label{LCwsum}
\end{equation}

Another coding mechanism that we will consider is called multiple description coding (MDC). For MDC, the video stream is partitioned into multiple substreams, referred to as descriptions. Any description can be used to decode the media stream. As the number of received substreams increases, the received video quality increases. For MDC, since there is no hierarchical layered coding system, each received unique description is equally beneficial to increase the video quality. For the MDC mechanism, $Q = Q_{MDC}$ and the chunk size $w_{j,q}$ is constant for any chunk $q$ for fixed video $j$, and the total size of the video $j$ with the highest video quality is equal to
\begin{equation}
w_j = Q_{MDC} w_{j,q}. \label{MDCwsum}
\end{equation}

Our interest is in users located in the plane over the area that is covered by the base stations, \ie uniformly distributed in $A_{cov} = \cup_{s \in \Theta} A_s.$ The probability of a user in the plane being covered by caches $s\in\Theta$ (and is not covered by additional caches) is denoted by $p_s = \vert A_s\vert / \vert A_{cov}\vert$. A user located in $A_s$, $s \in \Theta$ can connect to all caches in subset $s$, and has access to all video chunks stored in these caches.

We assume that there are $R$ video quality levels. The probability that video $j$ is requested with video quality $\rho$ by a user at $x_u \in s$ is denoted by $a_{j,\rho,s}$. Here $\rho = 1$ refers to the video with the lowest quality and $\rho = R$ refers to the case where the video has the highest quality. There is an obvious relation between the required chunks and the desired video quality. For the LC mechanism, files can be divided into $Q_{LC} = R$ layers with the proper size selections ($w_{j,q}$) and receiving all layers gives the video with the highest quality. On the other hand, for the MDC mechanism, in practice $Q_{MDC}$ is not equal to $R$ since all chunks have the equal size. We define a new parameter $D_\rho$ that represents the required number of chunks to get the video quality $\rho$.

We assume that the distribution for the requested video's quality depends on the user's location. To validate this assumption, one might think that users travelling in the bus will more likely request watching videos in a lower resolution compared to those located in residential areas. Similarly, users located in a residential area may request watching videos in different qualities as some users may watch videos via their ultra high-definition televisions, whereas others watch videos via their mobile phones with a lower resolution. We conclude that the requested video quality for a user located in $s\in\Theta$ follows a known and fixed probability mass function (pmf) $f^{(s)}(\rho)$.

For the file popularities, we assume that $a_1 \geq a_2 \geq \dots \geq a_J$ for any video quality $\rho$. We assume that the video popularities do not vary as $\rho$ changes. Even though any popularity distribution can be used, most of our numerical results will be based on the Zipf distribution for the video popularities for any video quality $\rho$. For any video quality indicator $\rho$, the probability that a user will ask for content $j$ is equal to
\begin{equation}
a_j = \frac{j^{-\gamma}}{\sum_{j=1}^J j^{-\gamma}}, \label{zipfpars}
\end{equation}
where $\gamma > 0$ is the Zipf parameter. 

Since $a_j$ and $f^{(s)}(\rho)$ represent statistically independent random variables, we conclude that the probability that video $j$ is requested by a user located at $x_u \in s$ with video quality $\rho$ is equal to 
\begin{equation}
a_{j,\rho,s} = \frac{j^{-\gamma}}{\sum_{j=1}^J j^{-\gamma}} f^{(s)}(\rho).
\end{equation}

Content is placed in caches using knowledge of the request statistics $a_{j,\rho,s}$, but without knowing the actual request made by the user.
We denote the placement policy for cache $m$ as
\begin{align}
b_{j,q}^{(m)} := \left\{
\begin{array}{rl}
1, & \text{if } c_{j,q} \text{ is stored in cache $m$},\\
0, & \text{if } c_{j,q} \text{ is not stored in cache $m$},
\end{array} \right.
\end{align}
and the overall placement strategy for cache $m$ as
\begin{equation*}
\mathbf{B}^{(m)} = 
\begin{bmatrix}
b_{1,1}^{(m)} & \ldots & b_{J,1}^{(m)}\\
\vdots & \ddots & \vdots\\
b_{1,Q}^{(m)} & \ldots & b_{J,Q}^{(m)}
\end{bmatrix}
\end{equation*}
as a $Q \times J$ matrix. 
The overall placement strategy for the network is denoted by $\mathbf{B} = \left[\mathbf{B}^{(1)}; \dots; \mathbf{B}^{(N)}\right]$ as an $Q \times J \times N$ three-dimensional matrix.

Caches have capacity $K$, \ie
\begin{equation*}
\sum_{j=1}^{J}\sum_{q=1}^Q w_{j,q} b_{j,q}^{(m)} \leq K, \forall m.
\end{equation*}
For clarity of presentation, we assume homogeneous capacity for the caches. However, our work can immediately be extended to the network topologies where caches have different capacities (\ie for the case where cache $m \in [1:N]$ has capacity $K_{m}$).

The user requests to watch one of the videos from the content library with certain quality. The aim is to place content in the caches ahead of time in order to minimize the residual bandwidth, \ie the cost for downloading the missing video chunks of the user's requested video with a certain video quality from the gateway (\ie the main server) in order to serve the users.

In the next subsections, we are interested in designing a cache placement strategy that minimizes the cost for downloading the user's requested video with a certain quality from the gateway for LC and MDC mechanisms.
\subsection{Layered coding mechanism}
For LC mechanism, the residual bandwidth $f_{LC}(\mathbf{B})$ is equal to
\begin{equation}
f_{LC}\left(\mathbf{B}\right) = \sum_{s \in \Theta}\sum_{j = 1}^J \sum_{\rho = 1}^R p_s a_{j,\rho,s} \left[\sum_{q=1}^\rho w_{j,q} \prod_{\ell \in s}\left(1 - b_{j,q}^{(\ell)}\right)\right].
\label{LCbw}
\end{equation}

Our goal is to find the optimal placement strategy minimizing the residual bandwidth as follows:
\begin{prb}
\label{LCprb}
\begin{align}
&\min \text{ } f_{LC}\left(\mathbf{B}\right)\nonumber\\
&\text{ }\mathbf{s.t.}\quad  \sum_{j = 1}^J \sum_{q = 1}^Q w_{j,q} b_{j,q}^{(m)} \leq K, \text{ }\forall m,\label{capacityconstraint}\\
&\hspace{0.9cm}b_{j,q}^{(m)} \in \{0,1\}, \text{ } \forall j, q, m.\label{sizeconstraint}
\end{align}
\end{prb}

It is easy to verify that Problem~\ref{LCprb} is not convex. We will provide a distributed asynchronous algorithm to address Problem~\ref{LCprb} in which we iteratively update the placement policy at each cache. In~\cite{AGS17} we have showed for a different performance measure that we can define an algorithm that can be viewed as the best response dynamics in a potential game. We will present a similar algorithm here. We make use of the following notation. Denote by $\mathbf{B}^{(-m)}$ the placement policies of all caches except cache $m$. We will write $f_{LC}(\mathbf{B}^{(m)},\mathbf{B}^{(-m)})$ to denote $f_{LC}\left(\mathbf{B}\right)$. Also, for the sake of simplicity for the potential game formulation that will be presented in the following section, let $f_{LC}^{(m)}$ denote the residual bandwidth of cache $m$, \ie
\begin{align}
f_{LC}^{(m)}\left(\mathbf{B}\right) &= \sum_{\substack{s \in \Theta \\ m \in s}} \sum_{j = 1}^J \sum_{\rho = 1}^R \sum_{q=1}^\rho p_s a_{j,\rho,s} w_{j,q} \left(1-b_{j,q}^{(m)}\right)
\prod_{\ell \in s \setminus \{m\}}\left(1 - b_{j,q}^{(\ell)}\right)\nonumber\\
&= \sum_{\substack{s \in \Theta \\ m \in s}} \sum_{j = 1}^J \sum_{\rho = 1}^R \sum_{q=1}^\rho p_s a_{j,\rho,s} w_{j,q} \left(1-b_{j,q}^{(m)}\right) \zeta^{(m)}\left(j, q\right),
\label{LClocbw}
\end{align}
where
\begin{equation} \label{eq:zeta}
\zeta^{(m)}\left(j, q\right) = \prod_{\ell \in s \setminus \{m\}}\left(1 - b_{j,q}^{(\ell)}\right).
\end{equation}

\subsection{Multiple description coding}
In MDC, receiving all $Q_{MDC}$ descriptions will provide the highest expected video quality experience for the users. For the sake of simplicity, we define a new parameter $y_{j,\rho,s}$ denoting the total number of descriptions that needs to be fetched from the backhaul in order to satisfy the expected video quality $\rho$ for video $j$ in $s\in\Theta$ as
\begin{equation}
y_{j,\rho,s} = \max\left\{0, D_\rho - Q + \sum_{q=1}^Q \prod_{\ell \in s}\left(1-b_{j,q}^{(\ell)}\right)\right\}.\label{yfuncstructure}
\end{equation}

For MDC, the residual bandwidth $f_{MDC}(\mathbf{B})$, \ie the cost for downloading the missing chunks of the user's requested video from the gateway is equal to

\begin{align}
f_{MDC}\left(\mathbf{B}\right) &= \frac{1}{Q}\sum_{s \in \Theta}\sum_{j = 1}^J \sum_{\rho = 1}^R p_s a_{j,\rho,s} w_j y_{j,\rho,s}.
\label{MDCbw}
\end{align}

Our goal is to find the optimal placement strategy minimizing the total number of descriptions that needs to be fetched from the backhaul, equivalently minimizing the residual bandwidth as follows:
\begin{prb}
\label{MDCprb}
\begin{align}
&\min \text{ } f_{MDC}\left(\mathbf{B}\right)\nonumber\\
&\text{ }\mathbf{s.t.}\quad  \sum_{j = 1}^J \sum_{q = 1}^Q w_j b_{j,q}^{(m)} \leq KQ, \text{ }\forall m,\label{MDCcapacityconstraint}\\
&\hspace{0.9cm}b_{j,q}^{(m)} \in \{0,1\}, \text{ } \forall j, q, m,\label{MDCsizeconstraint}.
\end{align}
\end{prb}

As in the LC mechanism, Problem~\ref{MDCprb} is not convex either. Therefore, we will provide a distributed asynchronous algorithm to address Problem~\ref{MDCprb} (similar to the procedure we followed for the layered mechanism) in which we iteratively update the placement policy at each cache. We will make use of the following notation. Denote by $\mathbf{B}^{(-m)}$ the placement policies of all caches except cache $m$. We will write $f_{MDC}(\mathbf{B}^{(m)},\mathbf{B}^{(-m)})$ to denote $f_{MDC}\left(\mathbf{B}\right)$. Also, for the sake of simplicity for the potential game formulation that will be presented in the following section, let $f_{MDC}^{(m)}$ denote the residual bandwidth of cache $m$, \ie
\begin{align}
f_{MDC}^{(m)}&\left(\mathbf{B}\right) = \frac{1}{Q}\sum_{\substack{s \in \Theta \\ m \in s}}\sum_{j = 1}^J \sum_{\rho = 1}^R p_s a_{j,\rho,s} w_j \max\left\{0, D_\rho - Q + \sum_{q=1}^Q \left(1-b_{j,q}^{(m)}\right)\zeta^{(m)}\left(j, q\right)\right\},
\label{MDClocbw}
\end{align}
where $\zeta^{(m)}\left(j, q\right)$ is given in~\eqref{eq:zeta}.

\section{Potential Game Formulation}
\label{sec:potentialgame}
In this section we provide a distributed asynchronous algorithm to address Problems~\ref{LCprb} and~\ref{MDCprb} in which we iteratively update the placement policy at each cache. We will show that this algorithm can be formulated as providing the best response dynamics in a potential game.

In our algorithm, each cache tries selfishly to minimize its residual bandwidth $f_{LC}^{(m)}$ and $f_{MDC}^{(m)}$ defined in~\eqref{LClocbw} and~\eqref{MDClocbw} for LC and MDC mechanisms respectively. Given a placement $\mathbf{B}^{(-m)}$ by the other caches, cache $m$ solves for $\mathbf{B}^{(m)}$ in

\begin{prb}
\label{prb:LCloc}
\begin{align}
&\min \text{ } f^{(m)}_{LC}\left(\mathbf{B}^{(m)},\mathbf{B}^{(-m)}\right)\nonumber\\
&\text{ }\mathbf{s.t.}\quad  \sum_{j = 1}^J \sum_{q = 1}^Q w_{j,q} b_{j,q}^{(m)} \leq K,\label{capacityconstraint2}\\
&\hspace{0.9cm}b_{j,q}^{(m)} \in [0,1], \text{ } \forall j, q,\label{sizeconstraint2}
\end{align}
\end{prb}
for LC mechanism and,

\begin{prb}
\label{prb:MDCloc}
\begin{align}
&\min \text{ } f_{MDC}^{(m)}\left(\mathbf{B}^{(m)},\mathbf{B}^{(-m)}\right)\nonumber\\
&\text{ }\mathbf{s.t.}\quad  \sum_{j = 1}^J \sum_{q = 1}^Q w_j b_{j,q}^{(m)} \leq KQ, \text{ }\forall m,\label{MDCcapacityconstraint2}\\
&\hspace{0.9cm}b_{j,q}^{(m)} \in [0,1], \text{ } \forall j, q,\label{MDCsizeconstraint2}
\end{align}
\end{prb}
for MDC mechanism.

Each cache continues to optimize its placement strategy until no further improvements can be made. At this point $\mathbf{B}$ is a \emph{Nash equilibrium strategy} that minimizes the overall residual bandwidth satisfying

\begin{equation}
f^{(m)}_{LC}\left(\mathbf{B}^{(m)},\mathbf{B}^{(-m)}\right) \leq f^{(m)}_{LC}\left(\bar{\mathbf{B}}^{(m)},\bar{\mathbf{B}}^{(-m)}\right), \text{ }\forall m, \mathbf{B}^{(m)},
\end{equation}
for LC mechanism and,
\begin{equation}
f^{(m)}_{MDC}\left(\mathbf{B}^{(m)},\mathbf{B}^{(-m)}\right) \leq f^{(m)}_{MDC}\left(\bar{\mathbf{B}}^{(m)},\bar{\mathbf{B}}^{(-m)}\right), \text{ }\forall m, \mathbf{B}^{(m)},
\end{equation}
for MDC mechanism.

For both coding mechanisms, we will refer to these games as the \emph{content placement games} and demonstrate in the next subsections that both games are potential games~\cite{MS96} with many nice properties.

\subsection{Convergence analysis}
In this section, we will prove for both coding mechanisms that if the caches repeatedly update their placement strategies it is guaranteed to converge to a Nash equilibrium in finite time. Note that the order of the updates is not important as long as all caches are scheduled infinitely often.

\begin{thm}
The content placement games defined by payoff functions~\eqref{LClocbw} and~\eqref{MDClocbw} are potential games with the potential function given in~\eqref{LCbw} and~\eqref{MDCbw}, respectively. For both games, if we schedule each cache infinitely often, the best response dynamics converges to a Nash equilibrium in finite time.
\end{thm}
\begin{proof}
In order to show that both games are potential with the potential functions $f_{LC}\left(\mathbf{B}\right)$ and $f_{MDC}\left(\mathbf{B}\right)$, it is easy to verify that
\begin{align*}
f^{(m)}_{LC}\left(\bar{\mathbf{B}}^{(m)},\mathbf{B}^{(-m)}\right) - f^{(m)}_{LC}\left(\mathbf{B}^{(m)},\mathbf{B}^{(-m)}\right)= f_{LC}\left(\bar{\mathbf{B}}^{(m)},\mathbf{B}^{(-m)}\right) - f_{LC}\left(\mathbf{B}^{(m)},\mathbf{B}^{(-m)}\right),
\end{align*}
and
\begin{align*}
f^{(m)}_{MDC}\left(\bar{\mathbf{B}}^{(m)},\mathbf{B}^{(-m)}\right) - f^{(m)}_{MDC}\left(\mathbf{B}^{(m)},\mathbf{B}^{(-m)}\right)= f_{MDC}\left(\bar{\mathbf{B}}^{(m)},\mathbf{B}^{(-m)}\right) - f_{MDC}\left(\mathbf{B}^{(m)},\mathbf{B}^{(-m)}\right),
\end{align*}
which completes the proof of the first statement. The proof states that the improvement in the residual bandwidth by the best response after each update is equal to the improvement in the residual bandwidth in the overall network. The detailed analysis is trivial and skipped due to space constraints. 

Now, since there exists only a finite number of placement strategies, none of the caches will be missed in the long-run. Moreover, each non-trivial best response provides a positive improvement in the potential function in a potential game. Hence, we are guaranteed to converge to a Nash equilibrium in finite time for both games.
\end{proof}

\subsection{Structure of the best response dynamics}
\label{subsec:BR}
In this subsection we will analyze the structure of the best response dynamics. We will show that solutions to Problem~\ref{prb:LCloc} and to Problem~\ref{prb:MDCloc} can be obtained by solving a linear optimization problem and we provide the solutions in closed form.

For notational purposes, let us first define the set $\Omega =\left\{(j,q),\left(j\in[1,J]\right), \left(q\in[1:Q]\right)\right\}$ that consists of all video index-video quality tuples.

\begin{thm}
\label{LCsol}
The optimal solution to Problem~\ref{prb:LCloc} is given by
\begin{align}
\label{LCoptsoleq}
\bar{b}^{(m)}_{j,q}  = \left\{
\begin{array}{rl}
1, & \text{if } \pi^{-1}_m(j,q) \leq K/w_{j,q},\\
0, & \text{if } \pi^{-1}_m(j,q) > K/w_{j,q},\\
\end{array} \right.
\end{align}
where $\pi_m: \Omega \rightarrow [1:J \times Q]$ satisfies
\begin{align*}
&\sum_{\substack{s \in \Theta \\ m \in s}} \sum_{\rho = \pi_m(1,2)}^R  p_s a_{\pi_m(1,1),\rho,s} w_{\pi_m(1,1),\pi_m(1,2)}\zeta^{(m)}\left(\pi_m(1,1),\pi_m(1,2)\right)\\
&\geq \sum_{\substack{s \in \Theta \\ m \in s}} \sum_{\rho = \pi_m(2,2)}^R  p_s a_{\pi_m(2,1),\rho,s} w_{\pi_m(2,1),\pi_m(2,2)}\zeta^{(m)}\left(\pi_m(2,1),\pi_m(2,2)\right)\\
&\geq \dots\\
&\geq \sum_{\substack{s \in \Theta \\ m \in s}} \sum_{\rho = \pi_m(J\times Q,2)}^R  p_s a_{\pi_m(J\times Q,1),\rho,s} w_{\pi_m(J\times Q,1),\pi_m(J\times Q,2)}\zeta^{(m)}\left(\pi_m(J\times Q,1),\pi_m(J\times Q,2)\right),
\end{align*}
where $\pi_m(i,1)$ is the file index $j$ of the i$^{th}$ sorted tuple, and $\pi_m(i,2)$ is the video quality index $q$ of the i$^{th}$ sorted tuple.
\end{thm} 
\begin{proof}
The objective function $f_{LC}^{(m)}\left(\mathbf{B}\right)$ is an affine function and we have several inequality constraints. It is easy to verify that we also have affine functions in the constraints. Therefore, we can conclude that Problem~\ref{prb:LCloc} is a linear problem. It can then be easily solved using dual problem. We have the Lagrange dual function
\begin{align*}
&\inf_{\mathbf{B}}\left(\sum_{\substack{s \in \Theta \\ m \in s}} \sum_{j = 1}^J \sum_{\rho = 1}^R \sum_{q=1}^\rho p_s a_{j,\rho,s} w_{j,q} \left(1-b_{j,q}^{(m)}\right) \zeta^{(m)}\left(j, q\right) +\lambda\left(\sum_{j = 1}^J \sum_{q = 1}^Q w_{j,q} b_{j,q}^{(m)} - K\right) - \sum_{j=1}^J \sum_{q = 1}^Q \mu_{j,q} b^{(m)}_{j,q} +\right.\\
&\left. \sum_{j=1}^J \sum_{q = 1}^Q \nu_{j,q} \left(b^{(m)}_{j,q} - 1\right)\right).
\end{align*}
For notational convenience we introduce the functions $\Psi_{j,q}: \mathbb{R}\rightarrow[0,1]$, $j = 1,\dots,J$, and $q = 1,\dots,Q$ as follows
\begin{align}
\label{LCPsi}
\Psi_{j,q}(\lambda)  = \left\{
\begin{array}{rl}
1,\hspace{0.1cm}\text{if } \lambda < \sum_{\substack{s \in \Theta \\ m \in s}} \sum_{\rho = q}^R  p_s a_{j,\rho,s} w_{j,q}\zeta^{(m)}\left(j,q\right),\\
0,\hspace{0.1cm}\text{if } \lambda \geq \sum_{\substack{s \in \Theta \\ m \in s}} \sum_{\rho = q}^R  p_s a_{j,\rho,s} w_{j,q}\zeta^{(m)}\left(j,q\right).\\
\end{array} \right.
\end{align}
We also define $\Psi: \mathbb{R}\rightarrow[0,K]$, where $\Psi(\lambda) = \sum_{j=1}^J \sum_{q=1}^Q w_{j,q} \Psi_{j,q}(\lambda) = K$ for 
$$
\lambda \in \left(-\infty,\sum_{\substack{s \in \Theta \\ m \in s}} \sum_{\rho = q}^R  p_s a_{j,\rho,s} w_{j,q}\zeta^{(m)}\left(j,q\right)\right),
$$
and $\Psi(\lambda) = \sum_{j=1}^J \sum_{q=1}^Q w_{j,q} \Psi_{j,q}(\lambda) = 0$ for
$$
\lambda \in \left[\sum_{\substack{s \in \Theta \\ m \in s}} \sum_{\rho = q}^R  p_s a_{j,\rho,s} w_{j,q}\zeta^{(m)}\left(j,q\right),\infty\right).
$$
It is then possible to check all possible combinations from the video content-video quality tuples $\Omega \rightarrow [1:J \times Q]$ to confirm if the condition given in~\eqref{LCPsi} is satisfied. In order to satisfy the capacity constraint~\eqref{capacityconstraint2} the above solution is guaranteed to exist. The proof is completed by validating that with the strategy above, $\bar{\lambda}$ solving the Lagrange dual is satisfying $\Psi(\bar{\lambda}) = K$.
\end{proof}
The intuition behind Theorem~\ref{LCsol} is as follows. The ordering is simply made based on the weights of the tuples consisting of the probability of being located in a specific region, the summation coming from the effect of the difference between requests with different video qualities (and the summation index difference due to the effect of the layered coding mechanism), and the effect of the possibility of receiving the chunks from the neighbours.

\begin{thm}
\label{MDCsol}
The optimal solution to Problem~\ref{prb:MDCloc} is given by
\begin{align}
\label{MDCoptsoleq}
\bar{b}^{(m)}_{j,q}  = \left\{
\begin{array}{rl}
1, & \text{if } \pi^{-1}_m(j,q) \leq KQ/w_j,\\
0, & \text{if } \pi^{-1}_m(j,q) > KQ/w_j,\\
\end{array} \right.
\end{align}
where $\pi_m: \Omega \rightarrow [1:J \times Q]$ satisfies
\begin{align*}
&\sum_{\substack{s \in \Theta \\ m \in s}} \sum_{\rho = 1}^R p_s a_{\pi_m(1,1),\rho, s} w_{\pi_m(1,1)} \zeta^{(m)}\left(\pi_m(1,1), \pi_m(1,2)\right)\mathbbm{1}\left(D_\rho - Q + \sum_{i = 1}^Q \zeta^{(m)}(\pi_m(1,1),i) > 0\right)\\
&\geq \sum_{\substack{s \in \Theta \\ m \in s}} \sum_{\rho = 1}^R p_s a_{\pi_m(2,1),\rho, s} w_{\pi_m(2,1)}\zeta^{(m)}\left(\pi_m(2,1), \pi_m(2,2)\right)\mathbbm{1}\left(D_\rho - Q + \sum_{i = 1}^Q \zeta^{(m)}(\pi_m(2,1),i) > 0\right)\\
&\geq\dots\\
&\geq \sum_{\substack{s \in \Theta \\ m \in s}} \sum_{\rho = 1}^R p_s a_{\pi_m(J\times Q,1),\rho, s} w_{\pi_m(J\times Q,1)}\zeta^{(m)}\left(\pi_m(J\times Q,1),\pi_m(J\times Q,2)\right)\\
&\mathbbm{1}\left(D_\rho - Q + \sum_{i = 1}^Q \zeta^{(m)}(\pi_m(J\times Q,1),i) > 0\right),
\end{align*}
where $\pi_m(i,1)$ is the file index $j$ of the i$^{th}$ sorted tuple, and $\pi_m(i,2)$ is the video quality index $q$ of the i$^{th}$ sorted tuple., and $\mathbbm{1}(.)$ is the indicator function, \ie it is equal to $1$ if the condition is satisfied and $0$ otherwise.
\end{thm} 
\begin{proof}
The objective function $f_{MDC}^{(m)}\left(\mathbf{B}\right)$ is a linear function and we have several inequality constraints. It is easy to verify that we also have linear functions in the constraints. Therefore, we can conclude that Problem \ref{prb:MDCloc} is a linear problem. It can then be easily solved using dual problem. We have the Lagrange dual function
\begin{align*}
&\inf_{\mathbf{B}}\left(\sum_{\substack{s \in \Theta \\ m \in s}}\sum_{j = 1}^J \sum_{\rho = 1}^R p_s a_{j,\rho,s} w_j\max\left\{0, D_\rho - Q + \sum_{q=1}^Q \left(1-b_{j,q}^{(m)}\right)\zeta^{(m)}\left(j, q\right)\right\}+\lambda\left(\sum_{j = 1}^J \sum_{q = 1}^Q w_{j} b_{j,q}^{(m)} - KQ\right)\right.\\ 
&\left.- \sum_{j=1}^J \sum_{q = 1}^Q \mu_{j,q} b^{(m)}_{j,q} + \sum_{j=1}^J \sum_{q = 1}^Q \nu_{j,q} \left(b^{(m)}_{j,q} - 1\right)\right).
\end{align*}
For notational convenience we introduce the functions $\Phi_{j,q}: \mathbb{R}\rightarrow[0,1]$, $j = 1,\dots,J$, and $q = 1,\dots,Q$ as follows
\begin{align}
\label{MDCPhi}
\Phi_{j,q}(\lambda)  = \left\{
\begin{array}{rl}
1, & \text{if } \lambda < \sum_{\substack{s \in \Theta \\ m \in s}} \sum_{\rho = q}^R  p_s a_{j,\rho,s} \zeta^{(m)}\left(j,q\right),\\
0, & \text{if } \lambda \geq \sum_{\substack{s \in \Theta \\ m \in s}} \sum_{\rho = q}^R  p_s a_{j,\rho,s} \zeta^{(m)}\left(j,q\right),\\
\end{array} \right.
\end{align}
when $D_\rho - Q + \sum_{i=1}^Q \zeta^{(m)}\left(j, i\right) > 0$.
We also define $\Phi: \mathbb{R}\rightarrow[0,K]$, where $\Phi(\lambda) = \sum_{j=1}^J \sum_{q=1}^Q w_j \Phi_{j,q}(\lambda)/Q = K$ for 
$$
\lambda \in \left(-\infty,\sum_{\substack{s \in \Theta \\ m \in s}} \sum_{\rho = q}^R  p_s a_{j,\rho,s} w_{j,q}\zeta^{(m)}\left(j,q\right)\right),
$$
and $\Phi(\lambda) = \sum_{j=1}^J \sum_{q=1}^Q w_j \Phi_{j,q}(\lambda) = 0$ for
$$
\lambda \in \left[\sum_{\substack{s \in \Theta \\ m \in s}} \sum_{\rho = q}^R  p_s a_{j,\rho,s} w_{j,q}\zeta^{(m)}\left(j,q\right),\infty\right).
$$
It is then possible to check all possible combinations from the video content-video quality tuples $\Omega \rightarrow [1:J \times Q]$ to confirm if the condition given in~\eqref{MDCPhi} is satisfied. In order to satisfy the capacity constraint~\eqref{MDCcapacityconstraint2} the above solution is guaranteed to exist. The proof is completed by validating that with the strategy above, $\bar{\lambda}$ solving the Lagrange dual is satisfying $\Phi(\bar{\lambda}) = K$.
\end{proof}

\subsection{Structure of Nash equilibria}
In this subsection we provide insight into the structure of the Nash equilibria of the content placement games. We know from the previous subsection that both games are potential games. Hence, the Nash equilibria for the games correspond to the optimal placement strategies satisfy the solutions of the dual problems of Problem~\ref{LCprb} and Problem~\ref{MDCprb}, respectively.

\begin{cor}
Let $\bar{\mathbf{B}}$ denote a placement strategy at a Nash equilibrium of the content placement game for LC mechanism. Then, $\bar{\mathbf{B}}$ satisfies~\eqref{LCoptsoleq} $\forall m = 1, \dots, N$.
\end{cor}
\begin{proof}
The proof is similar to the proof of Theorem~\ref{LCsol}, where in this case the optimal placement strategies must hold for all $1\leq m \leq N$ simultaneously.
\end{proof}

\begin{cor}
Let $\bar{\mathbf{B}}$ denote a placement strategy at a Nash equilibrium of the content placement game for MDC mechanism. Then, $\bar{\mathbf{B}}$ satisfies~\eqref{MDCoptsoleq} $\forall m = 1, \dots, N$.
\end{cor}
\begin{proof}
The proof is similar to the proof of Theorem~\ref{MDCsol}, where in this case the optimal placement strategies must hold for all $1\leq m \leq N$ simultaneously.
\end{proof}

\section{Numerical Evaluation}
\label{sec:performance}
In this section we will present practical implementations of our algorithm for the content placement games for LC and MDC mechanisms and evaluate our theoretical results according to a network of caches with their geographical locations following a real wireless network.
\subsection{The ROBR algorithm}
We will use the Random Order Best Response (ROBR) algorithm~\cite{AGS17} for the content placement games for both coding mechanisms to minimize the residual bandwidth. The basic idea of our algorithm is to repeatedly perform best response dynamics presented in Section~\ref{subsec:BR}. For ROBR algorithm, at each iteration step, a random cache is chosen uniformly from the set $[1:N]$ and updated by applying best response dynamics. We assume that all caches are initially empty (\ie $\mathbf{B}^{(m)} = 0_{Q,J}$, where $0_{Q,J}$ represents a $Q\times J$ zero matrix, $\forall m\in[1:N]$.). The algorithm stops when $f^{(m)}_{LC}\left(\mathbf{B}^{(m)},\mathbf{B}^{(-m)}\right)$ converges for LC mechanism, and $f^{(m)}_{MDC}\left(\mathbf{B}^{(m)},\mathbf{B}^{(-m)}\right)$ converges for MDC mechanism. ROBR algorithm is shown in Algorithm~\ref{alg:ROBR}.
\begin{figure}[!htb]
\removelatexerror
  \begin{algorithm}[H]
\label{alg:ROBR}
   \caption{Random Order Best Response (ROBR)}
   initialize $\mathbf{B}^{(m)} = 0_{Q,J}$, $\forall m \in [1:N]$\;
   set $imp(m) = 1$, $\forall m \in [1:N]$ \;
   set $\mathbf{imp} = [imp(1), \dots, imp(N)]$\;
   \While {$\mathbf{imp} \neq \mathbf{0}$}
   {
      m = Uniform(N)\;
      Set $imp(m) = 0$\;
      Solve Problem~\ref{prb:LCloc}* for cache $m$ and find $\mathbf{\bar{B}}^{(m)}$ using the information coming from neighbours\;
      Compute $f_{LC}^{(m)}(\mathbf{\bar{B}}^{(m)}, \mathbf{B}^{(-m)})$*\;
      \If{$f_{LC}^{(m)}(\mathbf{\bar{B}}^{(m)}, \mathbf{B}^{(-m)}) - f_{LC}^{(m)}(\mathbf{B}^{(m)}, \mathbf{B}^{(-m)}) \neq 0$*}
	{
		$imp(m) = 1$
	}
}
  \end{algorithm}
*For MDC mechanism, the algorithm solves Problem~\ref{prb:MDCloc} and uses the performance measure $f_{MDC}^{(m)}(\mathbf{\bar{B}}^{(m)}, \mathbf{B}^{(-m)})$.
\end{figure}

\subsection{A real wireless network: Berlin network}
In this section we will evaluate our theoretical results for the topology of a real wireless network. We have taken the positions of the base stations provided by the OpenMobileNetwork project~\cite{openmobilenetwork}. The base stations are located in the area $1.95 \times 1.74\text{ }km$s around the TU-Berlin campus. We will consider these base stations as our caches with certain cache capacities. The coverage radius of the base stations is equal to $r = 700\text{ }m$. The locations of the base stations is shown in Figure~\ref{fig:BerlinMap}.

\begin{figure}
\centering
\includegraphics[width=0.5\columnwidth]{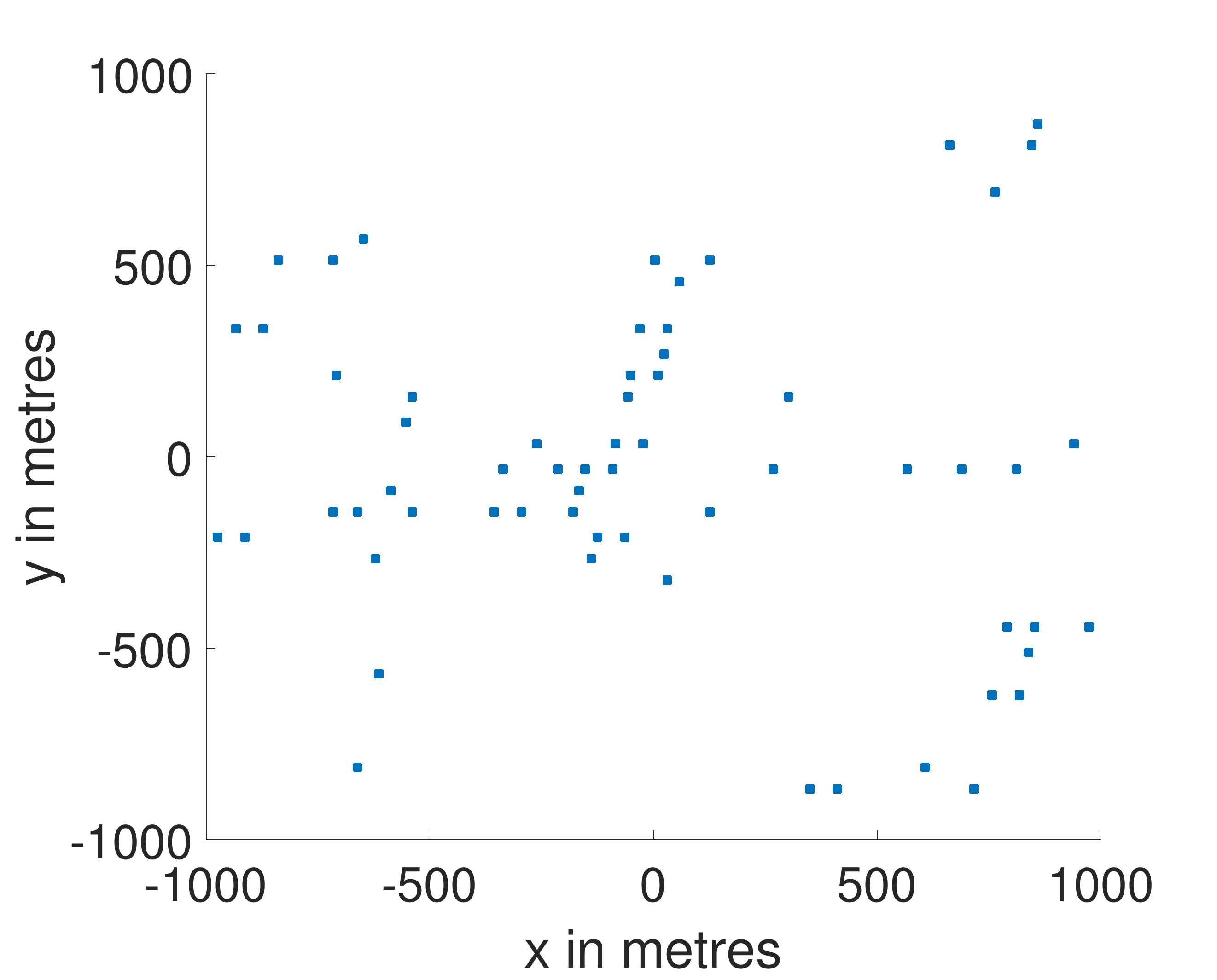}
\caption{Locations of Base Stations from OpenMobileNetwork dataset.}
\label{fig:BerlinMap}
\end{figure}

We consider the content library of size $J = 200$. We assume a Zipf distribution for the video popularities, setting $\gamma = 1$ and taking $a_j$ according to~\eqref{zipfpars}. We consider the case where videos have $R = 5$ different video qualities. We use the video quality bandwidth requirements and the base size data given in~\cite{KBZ13} for the chunk sizes, where the base size is the average of the sizes of a large set of tracked videos in certain video qualities given in megabytes per minute (MB/min). We assume that videos in the content library are all 40 minutes long. The corresponding video quality specifications is shown in Table~\ref{tab:spec}. The first column shows the video quality index $\rho$, the second column shows the base sizes of the videos and the third column shows the sizes of the 40 minutes long videos with different video qualities.

\begin{table}
\caption{Video quality specifications}
\centering
    \begin{tabular}{ | c | c | c |}
    \hline
$\begin{matrix}\text{Video quality}\\(\rho)\end{matrix}$& $\begin{matrix}\text{Base Size}\\\text{(MB/min)}\end{matrix}$&$\begin{matrix}\text{Video Size}\\\text{(GB)}\end{matrix}$\\
\hline
1 (240p)&2.56&0.1024\\
\hline
2 (360p)&4.30&0.1720\\
\hline
3 (480p)&6.79&0.2716\\
\hline
4 (720p)&15.49&0.6196\\
\hline
5 (1080p)&32.70&1.3080\\
    \hline
    \end{tabular}
\label{tab:spec}
\end{table}

The videos are partitioned into $Q_{LC} = 5$ chunks by LC mechanism. From the third column of Table~\ref{tab:spec}, it immediately follows that $w_{j,1} = 102.4$ MB, $w_{j,2} = 69.6$ MB, $w_{j,3} = 99.6$ MB, $w_{j,4} = 348.0$ MB, and $w_{j,5} = 688.4$ MB, $\forall j$. By MDC mechanism, the videos are partitioned into $Q_{MDC} = 38$ chunks and $w_{j,q} = 34.42$ MB $\forall j,q$. Consequently, $D_1 = 3$, $D_2 = 5$, $D_3 = 8$, $D_4 = 18$, and $D_5 = 38$. Note that the selection of $Q_{MDC}$ is based on the least common multiple of the third column of Table~\ref{tab:spec}. Finally, we set $K = 6.54$ GB, \ie each cache can store five 1080p movies.

For the distribution of the requested video qualities, we have
\begin{align}
\label{eq:fsrho}
f^{(s)}(\rho)  = \left\{
\begin{array}{rl}
\frac{1+t_\rho\theta_s}{5+10\theta_s}, & \text{if } \rho\in{1,2,3,4,5},\\
0, & \text{otherwise},\\
\end{array} \right.
\end{align}
where $t_\rho$ is the parameter representing the direction of the increment of the requested video qualities and $\theta_s$ is the slope of the linear increment.

\begin{figure*}
\centering
\includegraphics[width=1\columnwidth]{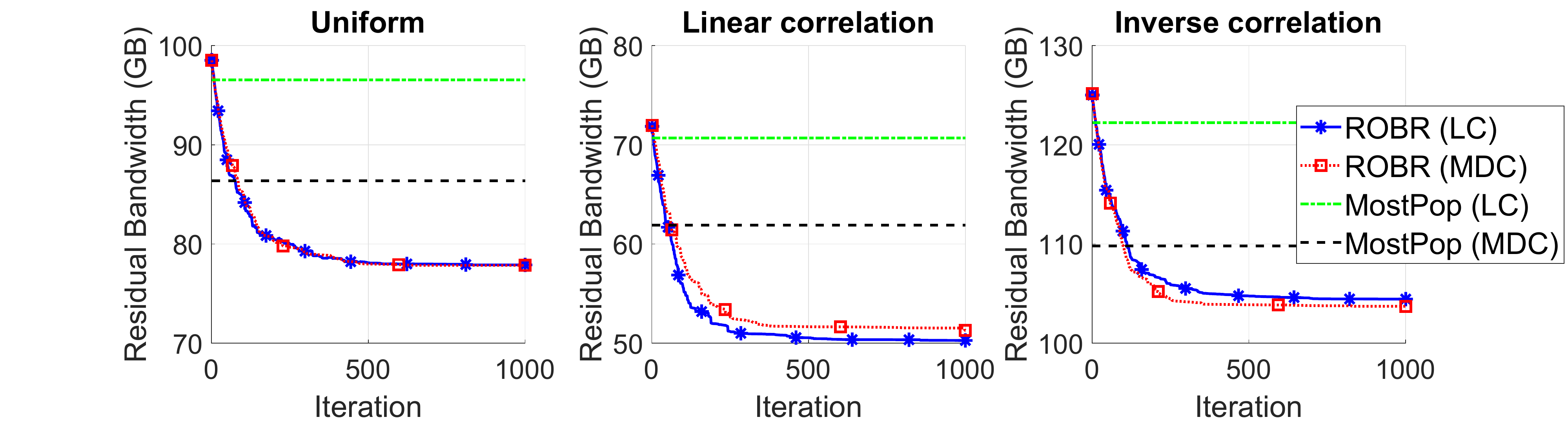}
\caption{The residual bandwidth evolution.}
\label{fig:RBws}
\end{figure*}

We will now evaluate our theoretical results for three different video quality trends.
\subsubsection{Uniform} For this trend, we set $\theta_s = 0$, $\forall s\in\Theta$. This trend reflects the case where users request videos in all qualities equally likely throughout the whole network.

\subsubsection{Linear correlation} For this trend, we set $t_\rho = 5 - \rho$ for $\vert s \vert < \mathbb{E}[\vert s \vert]$, and $t_\rho = \rho - 1$ for $\vert s \vert > \mathbb{E}[\vert s \vert]$, and $\theta_s = 0.1\left\vert\vert s \vert - \mathbb{E}[\vert s \vert]\right\vert$, where $\vert s \vert$ is the cardinality of the subset $s\in\Theta$, $\mathbb{E}[\vert s \vert]$ is the mean of the cardinalities of the all subsets, and $\left\vert\vert s \vert - \mathbb{E}[\vert s \vert]\right\vert$ is the absolute value of the difference of them. This trend reflects the case where users covered by many caches request videos in higher quality.

\subsubsection{Inverse correlation} We use the same parameters as in linear correlation case with the reversed $t_\rho$'s, \ie $t_\rho = \rho - 1$ for $\vert s \vert < \mathbb{E}[\vert s \vert]$, and $t_\rho = 5 - \rho$ for $\vert s \vert > \mathbb{E}[\vert s \vert]$. This trend reflects the case where users covered by many caches request videos in lower quality.

In Figure~\ref{fig:RBws} the residual bandwidth evolution for three different video quality trends for the proposed algorithm for both coding mechanisms is shown. In all video quality trends, storing the most popular content for MDC mechanism gives a lower residual bandwidth compared to LC mechanism since chunks have identical size and any received chunk increases the video quality for MDC. However, after running ROBR, the advantage of using MDC vanishes as both coding mechanisms give the same residual bandwidth performance. One might easily think that MDC would perform better than LC for a VoD system since any received chunk would increase the video quality. However, we see that LC can perform as good as MDC by applying ROBR and placing the chunks optimally. This is a crucial observation due to following reason. From technical point of view, at a given bitrate budget, there is a quality loss in MDC with respect to LC and chunks are more difficult to encode~\cite{V07}. In fact~\cite{WRL05} shows that some serious portion of the bitrate budget is reserved for forward error correction in MDC and the actual residual bandwidth will be higher for MDC to obtain the exact same quality. Finally, the residual bandwidth performance is the best for the linear correlation case as the requested video quality increases as the number of caches that users can connect to increases.
%%%%%%%%%%%%%%%%%%%%%%%%%%%%%%%%%%%%%%%%%%%%%%%%%%%%%%%%%%%%%%%%%%%%%%%%%%%%%%%%%%%%
% %
% %
% %
%%%%%%%%%%%%%%%%%%%%%%%%%%%%%%%%%%%%%%%%%%%%%%%%%%%%%%%%%%%%%%%%%%%%%%%%%%%%%%%%%%%%
\section{Discussion and Conclusion}
\label{sec:discussion}
In this paper we have provided a low-complexity asynchronously distributed cooperative caching algorithm for VoD when there is communication only between caches with overlapping coverage areas. We have related our algorithm to a best response dynamics in a game and have shown with practical examples that it converges rapidly. We have demonstrated the residual bandwidth evolution on a real network and have shown that our algorithm diminishes the advantage of using MDC over LC mechanism. Due to technical details, MDC suffers from quality loss since the chunks are more difficult to encode. Hence, we conclude that using LC for video partitioning is better for a VoD system as long as ROBR is applied for the optimal placement of the video chunks.

In future work we will look into this problem from the users' perspective and focus on maximizing the utilities of the users.
%%%%%%%%%%%%%%%%%%%%%%%%%%%%%%%%%%%%%%%%%%%%%%%%%%%%%%%%%%%%%%%%%%%%%%%%%%%%%%%%%%%%
% %
% %
% %
%%%%%%%%%%%%%%%%%%%%%%%%%%%%%%%%%%%%%%%%%%%%%%%%%%%%%%%%%%%%%%%%%%%%%%%%%%%%%%%%%%%%

%%%%%%%%%%%%%%%%%%%%%%%%%%%%%%%%%%%%%%%%%%%%%%%%%%%%%%%%%%%%%%%%%%%%%%%%%%%%%%%%%%%%
% %
% %
% %
%%%%%%%%%%%%%%%%%%%%%%%%%%%%%%%%%%%%%%%%%%%%%%%%%%%%%%%%%%%%%%%%%%%%%%%%%%%%%%%%%%%%
%\appendix
%%%%%%%%%%%%%%%%%%%%%%%%%%%%%%%%%%%%%%%%%%%%%%%%%%%%%%%%%%%%%%%%%%%%%%%%%%%%%%%%%%%%
% %
% %
% %
%%%%%%%%%%%%%%%%%%%%%%%%%%%%%%%%%%%%%%%%%%%%%%%%%%%%%%%%%%%%%%%%%%%%%%%%%%%%%%%%%%%%

%%%%%%%%%%%%%%%%%%%%%%%%%%%%%%%%%%%%%%%%%%%%%%%%%%%%%%%%%%%%%%%%%%%%%%%%%%%%%%%%%%%%
% %
% %
% %
%%%%%%%%%%%%%%%%%%%%%%%%%%%%%%%%%%%%%%%%%%%%%%%%%%%%%%%%%%%%%%%%%%%%%%%%%%%%%%%%%%%%
\end{document}